\documentclass[a4paper,12pt, leqno]{article}

\usepackage{natbib,enumerate,graphicx,array}
\usepackage{amsmath}
\usepackage{amsthm}
\usepackage{amssymb}
\usepackage{hyperref}
\usepackage{amsfonts}
\usepackage{amscd}
\usepackage[utf8]{inputenc}
\usepackage{palatino}

\newtheorem{theorem}{Theorem}

\newtheorem{definition}[theorem]{Definition}

\newtheorem{lemma}[theorem]{Lemma}

\newtheorem{proposition}[theorem]{Proposition}
\theoremstyle{definition}

\theoremstyle{remark}
\newtheorem{example}[theorem]{Example}

\newcommand* \ba {\mathop{\rm ba}\nolimits }
\newcommand* \argmin {\mathop{\rm argmin}\nolimits }

\newcommand* \pba {\mathop{\rm pba}\nolimits }

\newcommand* \conv {\mathop{\rm conv}\nolimits }

\newcommand*  \dd {\mathop{\ \rm d}\nolimits}
\newcommand{\myenddefinition}{$\blacklozenge$}

\newcolumntype{x}[1]{>{\centering\arraybackslash}p{#1}}
\newcolumntype{C}[1]{>{\centering\let\newline\\\arraybackslash\hspace{0pt}}m{#1}}
\usepackage{setspace}

\def\R{\mathbb{R}}

\def\O{\Omega}
\def\o{\omega}
\def\r{\rho}

\def\:={\mathrel{\mathop:}=}
\title{Three Variations on  Money Pump, Common Prior, and Trade}

\author{Ziv Hellman\thanks{\ Department of Economics, Bar-Ilan University, ziv.hellman@biu.ac.il. Support by Israel Science Foundation grant 448/22 is gratefully acknowledged.} \, and Mikl\'os Pint\'er\thanks{\ Corresponding author: Corvinus Center for Operations Research, Corvinus University of Budapest, pmiklos@protonmail.com. Support by the Hungarian Scientific Research Fund under projects K 133882 and K 119930 is gratefully acknowledged.}}

\begin{document}

\maketitle

\begin{abstract}
We consider finite information structures, and quest for the answer of the question: What is the proper definition of prior?

In the single player setting we conclude that a probability distribution is a prior if it is disintegrable, because this definition excludes money pump.

In the multiplayer setting our analysis does not boil down to one proper notion of common prior (the multiplayer version of prior). The appropriate notion is a choice of the modeller in this setting. We consider three variants of money pump, each "defines" a notion of common prior.

Furthermore, we also consider three variants of trade, each correspond to one of the money pump variants, hence to one of the common prior variants.

\bigskip

\end{abstract}


\section{Introduction}

Which come first, priors or posteriors? The question seems prima facie ill posed.
The words prior and posterior in their plain meanings carry connotations of prior-before, posterior-after.
Furthermore, the standard paradigm in nearly all of the decision theory and game theory literature, and by extension the economics literature in general, posits a an arrow of time, moving from the ex ante stage to the interim and then ex post stages. 
As part and parcel of this, agents start with prior probability distributions and move on to posterior distributions by conditioning. 
Priors, as their name indicates, are prior; posteriors follow, derived from priors.

We present a case here for the opposite: posteriors are prior to priors.
One way of being persuaded on this point is to consider the following example.

\begin{example}
	Let the state space be $\O = \{\o_1, \o_2, \o_3, \o_4, \o_5\}$.
	Two players partition $\O$ differently: player 1 has partition $\Pi_1 = \{\{\o_1, \o_2, \o_3\},\{ \o_4, \o_5\}\}$, and player 2 has partition $\Pi_2 = \{\{\o_1, \o_2\}, \{\o_3, \o_4, \o_5\}\}$.
	
	If the players begin with a uniform common prior ascribing probability $1/5$ to each state, then they derive the two posteriors
	\begin{align}
		\label{posteriorExample}
		\r_1 &= \{(1/3, 1/3, 1/3), (1/2,1/2)\} \\ \nonumber
				\r_2 &= \{(1/2,1/2), (1/3,1/3,1/3)\}
	\end{align}
	The posterior type structure above is, by construction, derived from a common prior. 
	This fact has significant behavioural implications for the players with respect to betting, trade, and many other properties, as has amply been studied in the literature.
	
	Consider, however, the possibility that the players begin with two different priors, player 1 with prior $\pi_1 = (1/6,1/6,1/6,1/4,1/4)$ and player 2 with the uniform prior $\pi_2 = (1/5,1/5,1/5,1/5,1/5)$.
	
	In this case, if the two players derive their respective posteriors, they will arrive at exactly the same posterior type structure (\ref{posteriorExample}) above. The behavioural implications will also be identical to those mentioned above, even though the players did not begin with a common prior
	\hfill \myenddefinition
\end{example}

The example illustrates that what counts behaviourally is the posterior, not the common prior.
In greater detail: whether or not the players begin with a common prior is immaterial for their posterior behaviour.
Instead, starting with the posteriors, the question is whether or not the players could have derived their respective posteriors from a common prior.
It is in this sense that posteriors are prior to priors.

This motivates the main question we ask in this paper: Given a posterior type structure, what is the proper definition of prior? Which probability distributions ought to count as priors to a given posterior? Quite simply, what is a prior?

The answer for our question is not mathematical.
We instead apply decision theoretic, economic, and financial considerations in attempts to identify what is a prior.
In the single player setting the answer is clear. 
The existence of a money pump is widely accepted as indicating irrationality.
We therefore should not define a prior in such a way as to leave open the possibility of money pumps.
We show that a prior probability distribution (with respect to a given posterior) excludes the possibility of a money pump if and only if it is disintegrable. Therefore, in the single player setting we conclude that a prior is a disintegrable probability distribution.

In the multiplayer setting, where we focus on finding a proper definition of common prior, the answer is not unequivocal. In this setting we apply all three considerations, decision theory, economics and finance, to find an answer. We analyse three variants of common prior (decision theory considerations), trade (economic considerations) and money pump (finance considerations). 
Our answer is that the proper definition of prior depends on the purpose of the modeller, specifically which decision theoretic or economic or financial considerations are most relevant for her model. 
We provide a complete description of nine different notions, and show how each one relates to the others. With this `map' at hand a modeller can choose the notion that best fits her needs, and can directly trace the decision theoretic, economic and finance implications of each.

We consider finite information structures (see Definition \ref{def1}) in this paper, that is, information structures in which the state space is finite. First we consider the single player setting, and redefine the notions of conglomerability (see Definition \ref{def:conglomerability}) and disintegrability (see Definition \ref{def:disintegrability}) in our setting. We maintain the intuitions of these notions, while differing from the classical definitions used in e.g. \citet{Dubins1975}. In our setting disintegrability implies conglomerability but not vice versa (see Proposition \ref{thm1} and Example \ref{pl}), while in the setting of \citet{Dubins1975} the two notions are equivalent.

We then consider the notion of money pump (see Definition \ref{money pump}), 
regarding it as equivalent to the belief of an uniformed player (outsider)
who has the opportunity to exploit a ``sure win'' arbitrage situation in which
at every state of the world both the single player and the uniformed player expect certain gain. Then we show that a probability distribution is a money pump if and only if it is not disintegrable (see Theorem \ref{thm:noMoneyPump0}). We conclude that a proper definition of prior is the one which excludes money pump, that is one that says that a probability distribution is a prior only if it is disintegrable (see Definition \ref{prior}).    
   
The multiplayer setting is more complicated than the single player case. First we introduce the notion of common certainty component (see Definition \ref{def:cbs}), which is a subset of the states in the states space such that at each state every player believes that the component occurs. In other words, a common certainty component is a set that is commonly believed by the players at each its states. Given an information structure and a common certainty component, the given information structure induces an information structure on the common certainty component (see Proposition \ref{prop:sub}). Therefore, each common certainty component can be regarded as a self-contained, independent (from the other parts of the state space) information structure.  

We use the three considerations mentioned above -- decision theory (common priors), finance (money pumps), and economics (trade) -- one at time, applying each of them to variants of common certainty components. 
This yields nine notions in all, three variants of common prior: common prior (see Definition \ref{def:cp1}), universal common prior (see Definition \ref{def:cp2}), and strong common prior (see Definition \ref{def:cp3}); trade: agreeable trade (see Definition \ref{def:trade2}), weakly agreeable trade (see Definition \ref{def:trade4}), and acceptable trade (see Definition \ref{def:trade5}); and money pump: multiplayer money pump (see Definition \ref{money pump2}), universal multiplayer money pump (see Definition \ref{money pump3}), and strong multiplayer money pump (see Definition \ref{money pump4}).  

Then we consider the relations of the above mentioned nine notions to each other.
This yields no trade theorems\footnote{It is worth to mention that although many often regard no trade and no betting as synonyms, they are not exactly the same, see \citet{GizatulinaHellman2019}.} (Theorems \ref{thm:notrade1}, \ref{thm:notrade2}, and \ref{thm:notrade3}) and no money pump theorems (Theorems \ref{thm:mpcp}, \ref{thm:mpcp2}, and \ref{thm:mpcp3}).
With these theorems at hand a modeller can choose which notion best fits a model in mind. 

In game theory no trade results typically state that an information structure attains a common prior if and only if there is no agreeable trade (bet) over it. This type of result was first presented%
\footnote{\ For a personal overview on the history of no trade results see \citet{Morris2020}.}
in Morris' PhD thesis \citep[Chapter 2]{Morris1991}, later published in \citet{Morris1994}.
Independently,  Yossi Feinberg presented an equivalent result in a working paper \citep{Feinberg1995} and in his PhD Thesis \citep{Feinberg1996} later published in \citet{Feinberg2000}. Both results relate to finite information structures; Morris'  proof is based on the Farkas' lemma \citep{Farkas1902}, while Feinberg's proof relies on the minimax theorem \citep{Neumann1928}. The very elegant paper by \citet{Samet1998} clarifies that at root both results above are based on the Strong Separating Hyperplane Theorem (see e.g. \citet{AliprantisBorder2006} Theorem 5.79).

\citet{Feinberg2000} no betting results consider not only finite information structures but also infinite compact information structures too. For no trade theorems on infinite information structures, see e.g. \citet{Heifetz2006,Hellman2014,LehrerSamet2014,HellmanPinter2022}.

The common prior notions used in \citet{Morris1991,Morris1994,Morris2020} on one side, and the one used in \citet{Feinberg1995,Feinberg1996,Samet1998,Feinberg2000} on the other side are not perfectly the same. We term the one used by \citet{Morris1991,Morris1994,Morris2020} a strong common prior in this paper, while the one by  \citet{Feinberg1995,Feinberg1996,Samet1998,Feinberg2000} we call common prior.

Some of the trade notions in this paper are not new either. The notion of agreeable trade (also called agreeable bet in the literature) was used by \citet{Samet1998}, while the notion of acceptable trade was used by \citet{Morris2020}. The notions of money pump we use in this paper are different in their forms from the ones typically appearing the literature. However, in spirit, those reflect the very same intuition : "sure win".

Hence some of the no trade theorems in this paper are not entirely new. In particular, Theorem \ref{thm:notrade1} can be found in \citet{Samet1998}, and Theorem \ref{thm:notrade2} is by \citet{Morris2020} (however, we give a completely new proof for this result). Regarding our no money pump theorems, those are new in their forms, but at the level of intuition, Theorem \ref{thm:mpcp3} can be found in \citet{Morris2020}.

The setup of the paper is as follows.  In Section \ref{sec:is} we discuss the notion of information structure we use in this paper. In Section \ref{sec:sp} we consider the single player case, and in Section \ref{sec:mpc} the multiplayer case. The last section briefly concludes.

\section{Information structure}\label{sec:is}

Let $N = \{1,\ldots,n\}$ be a set of players and $\Omega=\{\omega_1,\dots,\omega_M\}$ be a set of states of the world.
Each player $i$ is associated with a partition $\Pi_i$ of $\O$, which we call that player's knowledge partition.
We denote by $\Delta (\Omega)$ the collection of probability measures over $\O$, which we identify with the elements of the unit simplex in $\mathbb{R}^M$.

In this section we generally follow notions as they appear in \citet{Samet1998}. 

\begin{definition}
	A \emph{type function} of a player $i \in N$ is a function $t_i \colon \Omega \to \Delta (\Omega)$ that associates to each state $\o$ a distribution in $\Delta(\O)$, the \emph{type} $t_i(\o)$ of player $i$ at that state.
\end{definition}

\begin{definition}\label{def1}
An information structure is a tuple $(\Omega,(\Pi_i , t_i)_{i \in N})$, where $\Pi_i$, player $i$'s \emph{knowledge partition}, is a partition of $\O$, and $t_i \colon \Omega \to \Delta (\Omega)$ is the type function of player $i$. 
We assume that the type function $t_i$ of each player $i$ in an information structure satisfies the following two conditions: 

\begin{itemize}
\item $t_i (\omega)(\pi_i (\omega))=1$ for each $\omega \in \Omega$, where $\pi_i (\omega)$ is the element of the partition $\Pi_i$ which contains $\omega$, that is, $\omega \in  \pi_i (\omega)$,

\item $t_i$ is constant over each element of $\Pi_i$.
\end{itemize}
\end{definition}

In words, an information structure describes each player's knowledge and beliefs. 
For each player $i$, the knowledge partition $\Pi_i$ specifies which states the player can distinguish and which she cannot at each state $\o$. The type function $t_i (\o)$ expresses the belief of player $i$ at state $\o$, in the form of a probability distribution over the state space $\O$. 

Sometimes in the literature, the elements of $\Pi_i$ are referred to as the types of player $i$. 
We will refer to both the probability distribution $t_i (\o)$ and the element of the knowledge partition $\pi_i \in \Pi_i$ as types.
Furthermore, the type functions are also called posteriors, that is, $t_i$ is the \emph{posterior} of player $i$. 

\section{The Single Player Case}\label{sec:sp}

The standard, traditional view in the literature dictates a particular conceptual ordering: one starts with a prior distribution -- a measure over the state space -- from which a posterior distribution is subsequently defined by way of Bayesian updating (relative to a partition). This ordering is inherent in the names -- prior vs posteriors -- of the concepts themselves.

Here, however, we reverse that ordering and ponder instead the following question: given a posterior, what would constitute a prior for that posterior? 
We want the answer to be founded on decision theoretic, finance theoretic, and economics considerations, especially on one of the most fundamental principles in the literature in these fields, namely that money pumps and arbitrage opportunities are to be excluded.


In this section we consider the single player case. Specifically, we work with an information structure $(\Omega,(\Pi_i,t_i)_{i \in N})$ in which the player set $N$ is a singleton, that is, the information structure is $(\Omega,\Pi,t)$. 

We introduce here a variant of the concept of conglomerability. The notion of conglomerability was introduced by \citet{DeFinetti1930,deFinetti1972}. It is typically expressed as a condition imposed on the posterior of a player with respect to a given prior. Since in our model the posteriors are given and the objective is defining corresponding priors, we modify the concept, while maintaining the original intuition, as follows: 

\begin{definition}\label{def:conglomerability}
Let $(\Omega,\Pi,t)$ be a single player information structure. 
A probability distribution $p \in \Delta(\O)$ is \emph{$(\Pi,t)$-conglomerable} if for each $E \subseteq \O$
\begin{equation*}
\min \limits_{\omega \in \O} t (\omega) (E)  \leq p (E) \leq \max \limits_{\omega \in \O} t (\omega) (E) .
\end{equation*}
\end{definition}

In words, a probability distribution is conglomerable with respect to a single player information structure if the minimum and the maximum of the posterior of any event sandwiches the probability of the event. 
A natural case can be made to regard comglomerability as a desirable trait: it is counter-intuitive to conceive of a player ascribing unconditional probability to an event that is strictly greater than (respectively less than) the conditional probability of that same event at each and every possible state.

We also modify the standard definition of the concept of disintegrability in the literature to reflect our setting in which the posteriors are fixed and the objective is identifying priors.

\begin{definition}\label{def:disintegrability}
Let $(\Omega,\Pi,t)$ be a single player information structure.  
A probability distribution $p \in \Delta (\Omega)$ is \emph{$(\Pi,t)$-disintegrable} if 
\begin{equation*}
p (E \cap \pi) =  t (\omega) (E) \, p (\pi) ,
\end{equation*}
\noindent for every $E \subseteq \Omega$, every $\pi \in \Pi$, and all $\omega \in \pi$.
\end{definition}

In words, disintegrability is a generalisation of the familiar law of total probability. 
\citet{Samet1998} (Observation p. 172) demonstrated that a probability distribution is disintegrable if and only if it is located within the convex hull of the types. Formally,

\begin{lemma}\label{lem}
Let $(\Omega,\Pi,t)$ be a single player information structure. 
A probability distribution $p \in \Delta (\Omega)$ is $(\Pi,t)$-disintegrable if and only if $p \in \conv \{t (\o) \colon \o \in \O\}$.
\end{lemma}
 
\citet{Dubins1975} (pp. 90-91) showed the equivalence of disintegrability and conglomerability in the classical setting.
We show here, however, that this equivalence does not obtain in our setting, where disintegrability implies conglomerability but not vice versa.   

\begin{proposition}\label{thm1}
Let $(\Omega,\Pi,t)$ be a single player information structure.
If a probability distribution $p \in \Delta(\O)$ is $(\Pi,t)$-disintegrable then it is $(\Pi,t)$-conglomerable.
\end{proposition}

\begin{proof}
Suppose that $p$ is not $(\Pi,t)$-conglomerable, that is, that there exists $E \subseteq \O$ such that
\begin{equation*}
\min \limits_{\omega \in \O} t (\omega) (E) > p (E) \mspace{10mu} \text{or} \mspace{10mu} p(E) > \max \limits_{\omega \in \O} t (\omega) (E) .
\end{equation*}

\noindent Without loss of generality we may assume that $p(E) > \max_{\omega \in \O} t (\omega) (E)$. 

Then 

\begin{equation*}
\sum \limits_{\pi \in \Pi} t (\omega_\pi) (E) \, p (\pi) = \sum \limits_{\pi \in \Pi} p (E \cap \pi) = p (E) > \max \limits_{\omega \in \O} t (\omega) (E) ,
\end{equation*}

\noindent which is a contradiction, where $\o_\pi \in \pi$.
\end{proof}

To see that conglomerability does not necessarily imply disintegrability, consider the following example.

\begin{example}\label{pl}
Let $\O = \{\o_1,\o_2,\o_3\}$ be a state space with knowledge partition $\Pi = \{\{\o_1,\o_2\},\{\o_3\}\}$. Consider the type function 
\begin{equation*}
t (\o_1) = t (\o_2) = 
\begin{pmatrix}
0.9 \\
0.1 \\
0
\end{pmatrix} \mspace{10mu} \text{and} \mspace{10mu}
t (\o_3) = 
\begin{pmatrix}
0 \\
0 \\
1
\end{pmatrix} .
\end{equation*}

\noindent Moreover, let

\begin{equation*}
p = 
\begin{pmatrix}
0.1 \\
0 \\
0.9
\end{pmatrix} .
\end{equation*}

It is easy to check that $p$ is $(\Pi,t)$-conglomerable. However $p$ is not a convex combination of the types, hence by Lemma \ref{lem} $p$ is not $(\Pi,t)$-disintegrable.
\hfill \myenddefinition
\end{example} 

\subsection{Priors as a Form of Guarding Against Money Pumps}

In the classical theory of priors and posteriors, it is well known that an agent who fails to implement Bayesian updating for probability conditionalisation in deriving a posterior from a prior is susceptible to a diachronic Dutch book that can be used as a money pump. Here we apply a similar concept to seek a general definition of what a prior is, relative to an information structure.


\begin{definition}\label{money pump}
Let $(\Omega,\Pi,t)$ be a single player information structure. 
A \emph{semi-trade} is a function $f  \colon \Omega \to \R$ such that for all $\o \in \O$

\begin{equation*}
\int f \dd t (\o) \ge 0 .
\end{equation*}
\end{definition}

In words, in the case of a semi-trade the player expects non-negative gain at every state of the world. 

\begin{definition}\label{def:mp}
Given a single player information structure $(\Omega,\Pi,t)$. A probability distribution $p \in \Delta (\O)$ is a \emph{money pump} if there exists a semi-trade $f$ such that

\begin{equation*}
\int f \dd p  < 0 .
\end{equation*}
\end{definition}

The intuition here is that a function is a semi-trade if with respect to the posterior at each $\o$, the player believes that she cannot lose. A probability distribution $p$ is a money pump if there exists a semi-trade $f$ whose expectation with respect to $p$ is negative. Then one can see this as $p$ is the belief of an uninformed player, hence her belief is constant, and both the informed player and the uninformed player expect gain with the trade $(f,-f)$.

The existence of a money pump as defined here can easily be translated into the existence of a Dutch book in the standard sense in the literature, that is, the existence of a trade that an agent is willing to accept even though the agent is guaranteed to lose no matter what is the true state of the world. Suppose that $p$ and $f$ are as in Definition \ref{def:mp}. Then at a point in time before the true state of the world is revealed an agent would be willing to accept the trade $g = -f$ since in expectation $\int g \dd p  > 0 $.
However, \emph{a posteriori} at any state $\o$ such an agent would be guaranteed a loss by the property that $\int g \dd t (\o) < 0$.%
\footnote{\ Even if in Definition \ref{money pump} the expected gain is only non-positive, shifting the considered trade with a sufficiently small constant we get strictly positive gain while maintaining the money pump.} 

The following is a theorem of the alternative, stating that a probability charge is either disintegrability or a money pump, but not both.

\begin{theorem}\label{thm:noMoneyPump0}
Given a single player information structure $(\Omega,\Pi,t)$. Take an arbitrary probability distribution $p \in \Delta (\O)$. Then only one of the following holds:

\begin{enumerate}
\item $p$ is $(\Pi,t)$-disintegrable,

\item  $p$ is a money pump.
\end{enumerate}
\end{theorem}

\begin{proof}
Assume that $p$ is $(\Pi,t)$-disintegrable. Then by Lemma \ref{lem}, $p \in \conv \{t (\o) \colon \o \in \O\}$. Therefore, for every semi-trade $f$ we have $\int f \dd p \geq 0$, hence $p$ is not a money pump.

\bigskip

Now suppose that $p$ is a money pump. Then there exists a semi-trade $f$ such that $\int f \dd p  < 0$. Since for every $p' \in \conv \{t (\o) \colon \o \in \O\}$ and semi-trade $f$ we have $\int f \dd p' \geq 0$. Therefore, $p \notin \conv \{t (\o) \colon \o \in \O\}$, hence by Lemma \ref{lem}, $p$ is not $(\Pi,t)$-disintegrable.
\end{proof}

The following example shows that, in contrast to what Theorem \ref{thm:noMoneyPump0} concludes with regard to a disintegrable distribution, a conglomerable probability distribution can be a money pump.

\begin{example}\label{pl1}
Consider the information structure $(\Pi,t)$ and the probability distribution $p$ from Example \ref{pl}. Then $p$ is $(\Pi,t)$-conglomerable. 

Moreover, let 
\begin{equation*}
f = 
\begin{pmatrix}
-1  \\
9 \\
0
\end{pmatrix}.
\end{equation*}

Then $f$ is a non-negative valued trade. Indeed, $\int f \dd t(\o_1) = \int f \dd t(\o_2) = -0.9 + 0.9 = 0$, and $\int f \dd t(\o_3) = 0$.

Moreover, $p$ is a money pump, since $\int f \dd p = -0.1$. 
\hfill \myenddefinition
\end{example}

Returning to the question with which we started this section -- given a type function of a player $t: \O \to \Delta(\O)$, which probability distributions over $\O$ `ought' to be considered priors relative to $t$ -- we can now use the result of Theorem \ref{thm:noMoneyPump0} to obtain an answer.
Given the foundational assumption that a rationality requires avoiding money pumps, we conclude that the set of priors relative to $t$ are the probability distributions that are $(\Pi,t)$-disintegrable relative, because they are the ones that guard against the possibility of a money pump. 
Conglomerability,in contrast, does not suffice.


\begin{definition}\label{prior}
Let $(\Omega,\Pi,t)$ be a single player information structure. Then a probability distribution $p \in \Delta (\O)$ is a \emph{prior} (with respect to $t$) if it is $(\Pi,t)$-disintegrable.
\end{definition} 

Our definition of a prior is equivalent to the definition appearing in \citet{Samet1998}, where a prior is defined to be a distribution in the convex hull of the types.
However, to the best of our knowledge, this is the first time that a no money pump approach has been used to motivate the definition of a prior, given a posterior. 



\section{The multiplayer setting}\label{sec:mpc}

In this section we extend our focus to the multi-player setting, where the natural question becomes: given the posteriors of the players, what should count as a \emph{common prior}? The answer is not as straightforward as it may appear at first, because the proper notion of money pump is a more complex question in this case. Versions of the concept of money pump have analogues in parallel versions of concepts related to no trade theorems. We will therefore relate our various version of common prior concepts both to money pump theorems and no trade theorems.

We distinguish here three types of no trade theorems: 1) an agreeable trade version \citep{Samet1998,Feinberg2000,HellmanPinter2022}, 2) an agreeing to disagree version \citep{Aumann1976,Contreras-TejadaScarpaKubickiBrandenburgerLaMura2021}, and 3) the one by \citet{Morris2020}. 

In our consideration of these three types of common prior, trade and money pump, we redefine the applied notions of trade from the literature in an epistemic framework, and introduce completely new notions of common prior, trade and money pump.  

\begin{figure}[h!]
\centering
\includegraphics[width=0.8\textwidth]{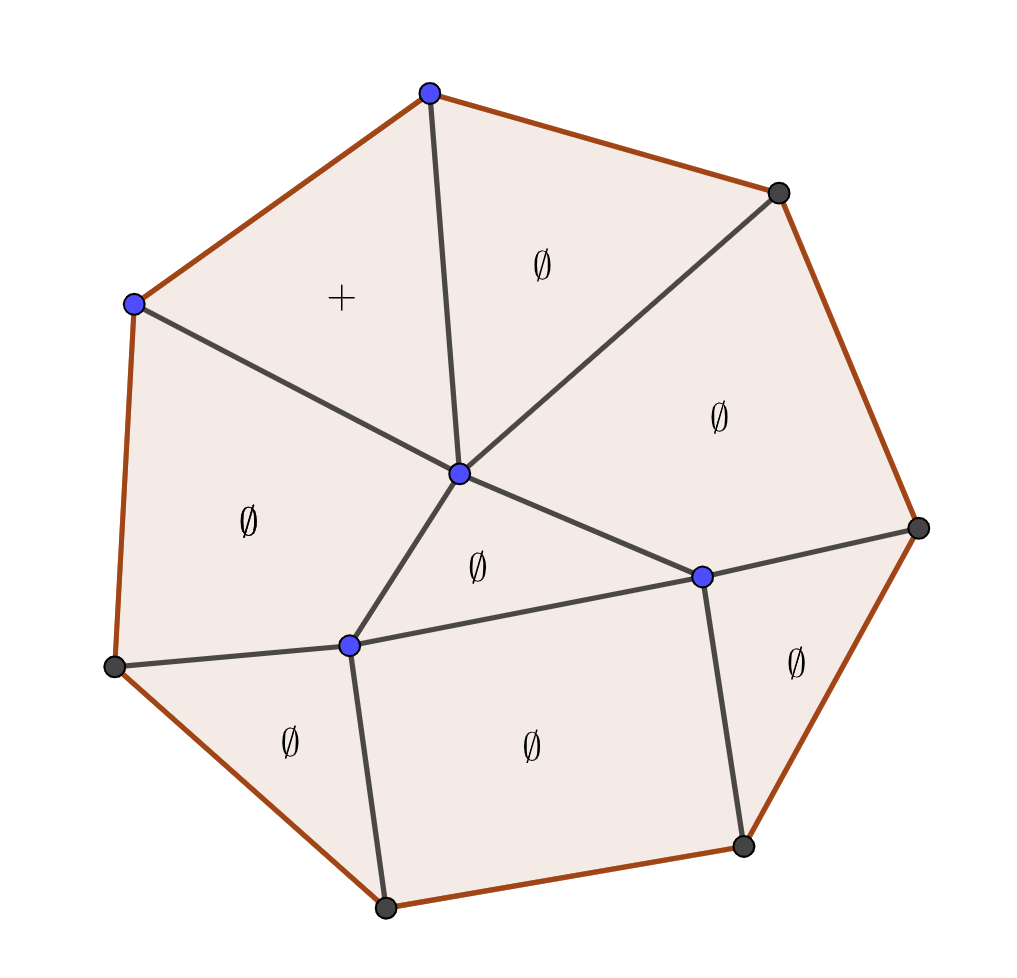}
\caption{Common Prior}\label{kep1}
\end{figure}

We can illustrate the differences among the three notions of common prior graphically. The basic question is what is going on in the common certainty components of the state space (precise definitions of these concepts will appear later; at this point we are only presenting an intuitive picture). 

A common certainty component is a subset of the state space which, from the decision theory perspective, is independent of the other parts of the state space. In other words, we may analyse each common certainty component as a separate and self-contained information structure. Loosely speaking, the common certainty components form a partition of the state space.

In Figure \ref{kep1} the polygons represent the common certainty components of the state space. A plus sign (respectively $\emptyset$) indicates that the given common certainty component attains (respectively, does not attain) a common prior in the sense of \citet{Aumann1976,Samet1998,Feinberg2000,HellmanPinter2022}. Under this version, a common prior exists over the entire space if and only if there exists at least one common certainty component that attains a common prior. 

In Figure \ref{kep1}, although only one common certainty component attains a common prior, that is sufficient for the entire information structure to attain a common prior. Moreover, notice that a common prior is not necessarily meaningful on certain common certainty components, e.g. it is meaningless on a common certainty component whose probability is zero by the common prior. 
 
\begin{figure}[h!]
\centering
\includegraphics[width=0.8\textwidth]{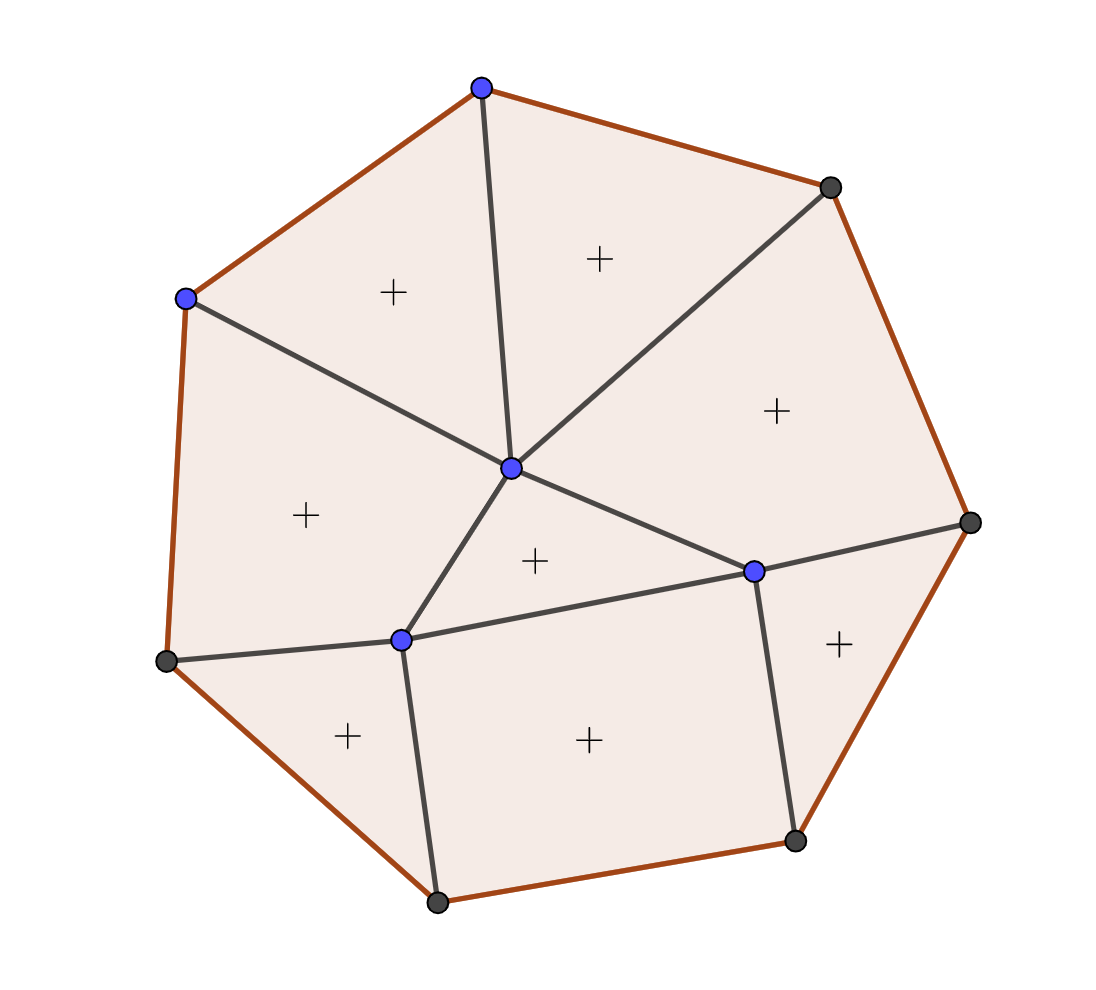}
\caption{Universal Common Prior}\label{kep2}
\end{figure}

Figure \ref{kep2} represents a case in which each and every common certainty component attains a common prior: this corresponds to the existence of a universal common prior. Clearly this notion is stronger than the one depicted in the previous figure. A universal common prior is meaningful on every common belief subspace, this is the difference between common prior and universal common prior. 

\begin{figure}[h!]
\centering
\includegraphics[width=0.8\textwidth]{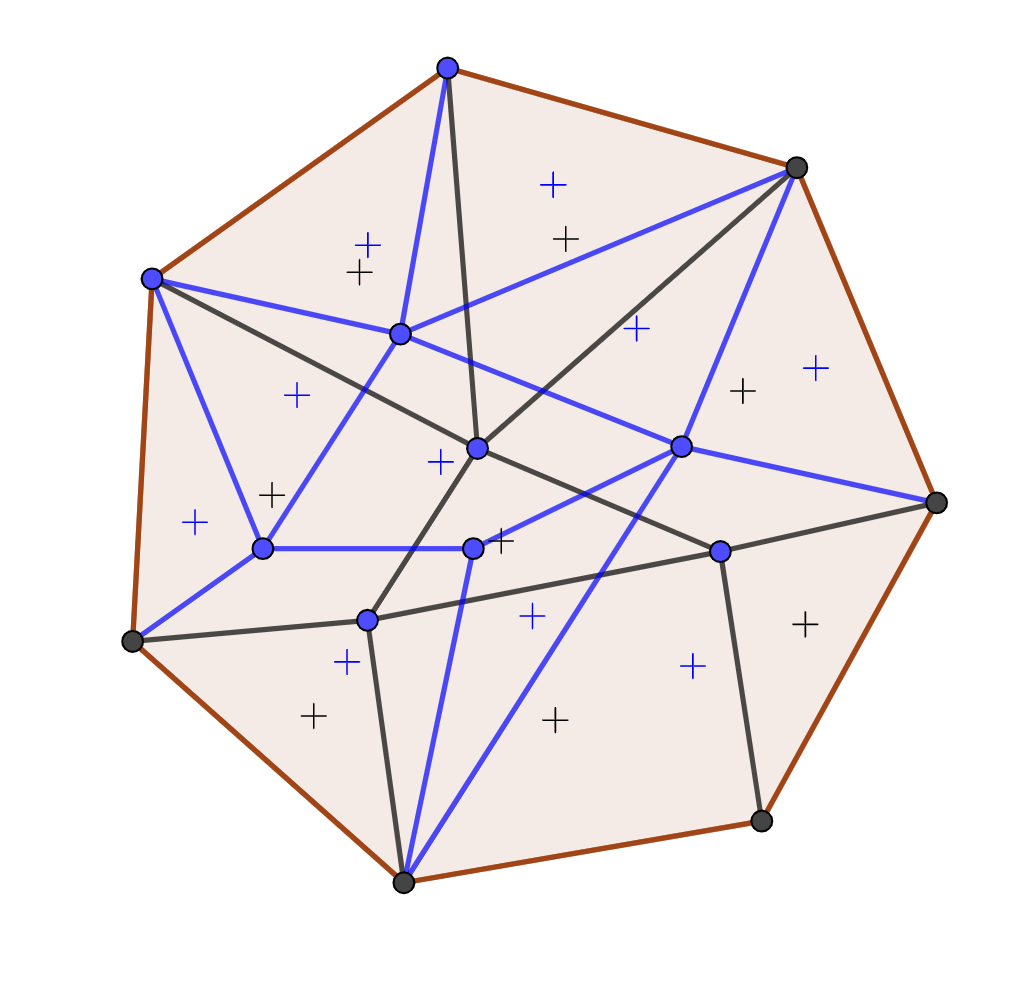}
\caption{Strong Common Prior}\label{kep3}
\end{figure}

In Figure \ref{kep3} we illustrate a situation where there are two players, labelled black and blue. The coloured polygons represent the knowledge partitions of the players (therefore, here the polygons are not common certainty components but elements of the knowledge partitions of the players). A strong common prior is a common prior such that it assigns positive probability to every player's every element of her knowledge partition (every type). It is clear that strong common prior is a stronger notion than universal common prior. 

The relation among the three considered notions as follows: if an information structure attains a strong common prior then it attains a universal common prior, but not vice versa. Furthermore, if an information structure attains a universal common prior then it attains a common prior, but not vice versa.

Our three notions of common prior, which we derive from decision theory considerations, have counterparts based on economic (trade) and finance (money pump) considerations.
Hence we will formalise three no trade theorems (common prior vs. trade) and three no money pump theorems (common prior vs. money pump).


Finally, throughout this section we fix an information structure $T = (\Omega,(\Pi_i,t_i)_{i \in N})$. Furthermore, for each player $i$ denote by $P_i$ the set of priors, according to Definition \ref{prior}). In other words,
\begin{equation*}
P_i = \left\{ p \in \Delta (\O) \colon p (E \cap \pi_i) =  t_i (\o) (E) \, p (\pi_i), \mspace{10mu} \forall E \subseteq \O,\ \pi_i \in \Pi_i \right\}.
\end{equation*}

\subsection{Common Certainty Components}

A common certainty component $S \subseteq \Omega$ is a set of states of the world such that at each state of the world in $S$ every player believes with probability one that $S$ occurs. This notion is inspired by the notion of common certainty in \citet{Contreras-TejadaScarpaKubickiBrandenburgerLaMura2021}, and it is different from the notion called common $1$-belief in \citet{MondererSamet1989}.
Formally, 

\begin{definition}\label{def:cbs}
A non-empty set $S \subseteq \O$ is a common certainty component if for every $\o \in S$ and $i \in N$ there exists $E \subseteq \O$ such that

\begin{itemize}
\item $E \subseteq S$,

\item $t_i (\o) (E) = 1$.
\end{itemize}
\end{definition}

Another way of describing a common certainty component $S$ is saying that at each state in $S$ every player can exclude with probability one the states which are not in $S$. 

Each common certainty component can be considered to be an information structure itself:

\begin{proposition}\label{prop:sub}
Given a common certainty component $S \subseteq \Omega$, let

\begin{itemize}
\item $\Pi_i^S \:= \{ E \subseteq S \colon \text{ there exists } \pi \in \Pi_i \text{ such that } E = \pi \cap S\}$, for all $i \in N$,

\item $t_i^S (\o) \:= t_i (\o)|_S$, for all $\o \in S$ and $i \in N$.
\end{itemize}

Then $(S,(\Pi_i^S,t_i^S)_{i \in N})$ is an information structure.
\end{proposition}

\begin{proof}
The only thing to prove is that $t_i^S$ meets Definition \ref{def1}. Suppose by contradiction that there exist a player $i \in N$ and a state $\o \in S$ and events $E,F \subseteq \O$ such that $S \cap E = S \cap F$ and $t_i (\o) (E) \neq t_i (\o) (F)$.

Since $S \cap E = S \cap F$ we have that $(E \setminus F) \cup (F \setminus E) \subseteq S^\complement$. Moreover, $S$ is a common certainty component, hence there exists $K \subseteq S$ such that $t_i (\o) (K) = 1$, meaning that $t_i (\o) ((E \setminus F) \cup (F \setminus E)) = 0$ which is a contradiction.  
\end{proof}

It is clear that an empty set cannot be a common certainty component, hence the intersection of two common certainty components is not necessarily a common certainty component. However, it is easy to check that the union of common certainty components is a common certainty component.

\begin{lemma}
The union of common certainty components is a common certainty component, but the intersection of common certainty components is not necessarily a common certainty component.
\end{lemma}  

The following example demonstrates that the complement of a common certainty component is not necessarily a common certainty component either. 

\begin{example}\label{ex:pl1}
Consider the information structure in Figure \ref{fig:fig1}. Here there are four common certainty components: $\{\o_1\}$, $\{\o_4\}$, $\{\o_1,\o_4\}$ and $\O$.

%


\begin{figure}[h!]
	\centering
{
	\begin{tabular}{l|c c c c c c c c }
		\hbox{Anne} &  1   &  \vline  &  $\frac{1}{2}$ &  &  $\frac{1}{2}$ & \vline & $1$ &  \vline  		
		\\ \hline
		\hbox{Ben } &  $1$ &  & $0$ & \vline  & $0$  & &  $1$ & \vline  \\ 
	\end{tabular}
}
	\caption{\small{The informations structure of Example \ref{ex:pl1}.}}\label{fig:fig1}
\end{figure}

It is clear that the complement of any common certainty component is not a common certainty component.
\end{example}

\subsection{Decision Theory Considerations: Common Prior}

In this subsection we introduce three variants of common prior that have appeared in  the literature. 


The first notion of common prior we identify is the one used by \citet{Samet1998}, which is the closest to the plain meaning of `common prior': a probability distribution that is simultaneously a prior from the perspective of all the players. This is the weakest notion amongst the notions under consideration here.

\begin{definition}\label{def:cp1}
A probability charge $p \in \Delta (\O)$ is a common prior if it is a prior for every player, that is, if $p \in \cap_{i \in N} P_i$.
\end{definition}


In the next example we present an information structure which does not attain a common prior.

\begin{example}[An information structure without a common prior]\label{ex:pl2}
Let the player set be $N = \{1,2\}$, the state space be $\Omega = \{\o_1,\o_2,\o_3,\o_4\}$. Consider two partitions:
\begin{eqnarray*}
\Pi_1=&\{\{\omega_1,\omega_2\}\{\omega_3,\omega_4\}\}, \\
\Pi_2=&\{\{\omega_1,\omega_4\}\{\omega_2,\omega_3\}\}.
\end{eqnarray*}

Let the type functions be
\begin{eqnarray*}
t_1(\omega_1)=&t_1(\omega_2)=&(1/2,\  1/2,\  0,\  0), \\
t_1(\omega_3)=&t_1(\omega_4)=&(0,\  0,\  1/2,\  1/2).
\end{eqnarray*}
\begin{eqnarray*}
t_2(\omega_1)=&t_2(\omega_4)=&(1/2,\  0,\  0,\  1/2), \\
t_2(\omega_2)=&t_2(\omega_3)=&(0,\  1,\  0,\  0).
\end{eqnarray*}

If $p$ were a common prior then $p (\{\o_1\}) = p (\{\o_2\}) = p (\{\o_3\}) = p (\{\o_4\})$, and $p (\{\o_3\}) = 0$, which is a contradiction. Hence this information structure does not attain a common prior.
\hfill \myenddefinition
\end{example}

In the following proposition we give some intuition motivating Definition \ref{def:cp1}. An information strucutre attains a common prior if there is a common certainty component on which the induced information strucutre attains a common prior.

\begin{proposition}\label{prop:prop1}
An information structure over $\O$ attains a common prior if and only if there exists a common certainty component $S \subseteq \O$ such that the induced information structure over $S$ attains a common prior.
\end{proposition}

First we need the following lemma. If the induced informations structure over a common certainty component attains a common prior, then by nullifying the probability of all events in the complement of the common certainty component we get a prior for the original informational structure. Formally:

\begin{lemma}\label{lem:lem1}
Suppose that $S$ is a common certainty component such that $T_S = (S,(\Pi_i^S),t_i^S)_{i \in N})$ attains a common prior. Then $T$ (the original informational structure) attains a common prior.
\end{lemma}

\begin{proof}
Let $S$ be a common certainty component and $p_S$ be a common prior for $T_S$. For each $E \subseteq \O$ let $p (E) = p_S (E \cap S)$; then $p \in \Delta (\O)$. 

Next we show that $p$ is a common prior for $T$. Let $i \in N$, $E \subseteq \O$, and $\pi_i \in \Pi_i$. Then
\begin{equation*}
p (E \cap \pi_i)  = p_S (E \cap \pi_i \cap S)  =  t_i^S (\o_{\pi_i^S}) (E \cap S) \, p_S (\pi_i^S) =  t_i (\o_{\pi}) (E) \, p  (\pi_i) , 
\end{equation*}

\noindent for all $\o_{\pi_i^S} \in \pi_i^S$ and $\o_{\pi_i} \in \pi_i$, where $\pi_i^S = \pi_i \cap S$.
\end{proof}

\begin{proof}[The proof of Proposition \ref{prop:prop1}]
Since $\O$ is a common certainty component, if $p$ is a common prior over $\O$ then it is a common prior for the informational strucutre $T$.

\bigskip

If $S$ is a common certainty component and $p_S$ is a common prior for $T_S$ then by Lemma \ref{lem:lem1} the information structure $T$ attains a common prior. 
\end{proof}


The following notion of common prior, that of a \emph{universal common prior}, has no antecedent in the literature. A universal common prior is a common prior that does not vanish on any common certainty component. This notion fits well with the setting of \citet{Aumann1976}, where common knowledge components must be assigned positive probability under the prior.

\begin{definition}\label{def:cp2}
An information structure $T$ admits a \emph{universal common prior} if there exists a common prior $p \in \Delta (\O)$ such that $p (S) > 0$ for every common certainty component.
\end{definition}


The following example shows that it is possible for an information strucutre to attain a common prior but not a universal common prior.  

\begin{example}[Universal common prior $\neq$ common prior] \label{pl4}
Let the player set be $N = \{1,2\}$, the state space be $\Omega = \{\o_1,\o_2,\o_3,\o_4\}$, and the two partitions be the same:
\begin{eqnarray*}
\Pi=&\{\{\omega_1,\omega_2\}\{\omega_3,\omega_4\}\} .
\end{eqnarray*}

Let the type functions be
\begin{eqnarray*}
t_1(\omega_1)=&t_1(\omega_2)=&(1/2,\  1/2,\  0,\  0), \\
t_1(\omega_3)=&t_1(\omega_4)=&(0,\  0,\  1/2,\  1/2).
\end{eqnarray*}
\begin{eqnarray*}
t_2(\omega_1)=&t_2(\omega_2)=&(1/2,\  1/2,\  0,\  0), \\
t_2(\omega_3)=&t_2(\omega_4)=&(0,\  0,\  1,\  0).
\end{eqnarray*}

Here $p = (1/2,\ 1/2,\ 0,\ 0)$ is the unique common prior, but the information structure does not admit a universal common prior since $p$ does not assign positive probability to the common certainty component $S = \{\o_3,\o_4\}$.
\hfill \myenddefinition
\end{example} 

An information structure admits a universal common prior if for each common certainty component the induced informational structure attains a common prior. Formally,

\begin{proposition}\label{prop:prop2}
An information structure $T$ admits a universal common prior if and only if for each common certainty component $S$ the information structure $T_S$ attains a common prior.
\end{proposition}

\begin{proof}
By definition, if $T$ admits a  universal common prior $p$, then for each common certainty component $S$ the probability distribution $p|_S$ is a common prior for the information structure $T_S$, where 
\begin{equation*}
p|_S (E) := \dfrac{p (E)}{p (S)} \,, 
\end{equation*}

\noindent $E \subseteq S$.

\bigskip

If for every common certainty component $S$ the induced information structure $T_S$ attains a common prior $p_S$, then by Lemma \ref{lem:lem1} $T$ also attains a common prior. Let $p^S$ denote these common priors. Since there are finite many common certainty components, any strict convex combination (every weight is positive) of $p^S$ is a universal common prior for the information structure $T$.
\end{proof}


The last notion of common prior that we consider first appeared in the literature in \citet{Morris2020}. The inspiration for this concept comes from noting that even when when an information structure admits a universal common prior it may happen that the posterior of a player at a particular state is `hidden' from the universal common prior, in the sense that it assigns zero probability to the type of the player. The following notion excludes this case. 

\begin{definition}\label{def:cp3}
An information strucutre $T$ \emph{admits a strong common prior} if there exists a common prior $p$ such that $p (\pi_i) > 0$ for every player $i \in N$ and type $\pi_i \in \Pi_i$.
\end{definition}


The next example shows that not every universal common prior is a strong common prior.

\begin{example}[Universal common prior $\neq$ strong common prior]\label{ex:pl3}
Consider the information structure in Example \ref{ex:pl1}. 

It is easy to see that the set of the common priors is $\conv \{\delta_{\o_1},\delta_{\o_4}\}$, the convex hull of $\delta_{\o_1}$ and $\delta_{\o_4}$. Therefore, every common prior assigns probability $0$ to type $2$ of Anne, ($\{\o_2,\o_3\}$).
Hence this information structure does not admit a strong common prior.

However, $\frac{1}{2}\delta_{\o_1} + \frac{1}{2} \delta_{\o_4}$ is a universal common prior since it assigns positive probability to each common certainty component $\{\o_1\}$, $\{\o_4\}$, $\{\o_1,\o_4\}$, and $\O$.
\hfill \myenddefinition
\end{example}

Directly from the definitions of the three notions of common prior introduced above it follows that they are related to each other in the following way:

\begin{proposition}\label{has1}
If an information structure admits a strong common prior then it admits a universal common prior. Moreover, if an information structure admits a universal common prior then it admits a common prior.
\end{proposition}

Notice that by Examples \ref{pl4} and \ref{ex:pl3} the statements above cannot be reversed, that is, not every common prior is a universal common prior, and not every universal common prior is a strong common prior.

\subsection{Economics Considerations: Trade}

In this subsection we consider three notions of trade (also called bet in the literature).
Each of them is an economic counterpart of a common prior from the previous subsection. Each notion expresses a scenario in which trade happens.

\begin{definition}\label{def:trade1}
A \emph{trade} $f=(f_{i})_{i \in N}$ is a family of functions $f_i$, $i \in N$, such that

\begin{equation*}
\sum_{i \in I} f_i(\omega)  \leq 0 .
\end{equation*}
\end{definition}

The following notion was introduced by \citet{Samet1998}. 
\begin{definition}\label{def:trade2}
A trade  $f=(f_{i})_{i \in N}$ is \emph{agreeable} if
\begin{equation*}
\int f_i \dd t_i (\o) > 0 .
\end{equation*}
for every state of the world $\o \in \O$ and every player $i \in N$.

\end{definition}

In words, a trade is agreeable if at every state of the world each player believes that she gains positive payoff. 

\begin{example}\label{ex:plbet1}
Consider the information structure from Example \ref{ex:pl2}. Then the pair $f_1 = (2,-1,4,-3)$ and $f_2 = - f_1$ is an agreeable trade.
\hfill \myenddefinition
\end{example}

We say that an event $E \subseteq \O$ is \emph{commonly certain} at a state of world $\o \in \O$ if there exists a common certainty component $S$ such that $\o \in S \subseteq E$. 
Another, equivalent definition of agreeable trade states that a trade is agreeable if at each state of world the players are commonly certain that each of them gains positive payoff.

\begin{definition}\label{def:trade3}
A trade  $f=(f_{i})_{i \in N}$ is \emph{agreeable} if the event $\{ \o' \in \O \colon \int f_i \dd t_i (\o') > 0, \ i \in N \}$ is commonly certain at every $\o \in \O$. 
\end{definition}

It is easy to see the following:

\begin{lemma}
Definitions \ref{def:trade2} and \ref{def:trade3} are equivalent.
\end{lemma}

We introduce another notion of trade: a trade is weakly agreeable if there exists at least one state of the world at which it is commonly certain that each of the players gains positive payoff.  

\begin{definition}\label{def:trade4}
A trade $f=(f_{i})_{i \in N}$ is \emph{weakly agreeable} if there exists a state of the world $\o \in \O$ such that at $\o$ the event $\{ \o' \in \O \colon \int f_i \dd t_i (\o') > 0, \ i \in N \}$ is commonly certain.
\end{definition}

Comparing the definitions of agreeable trade and weakly agreeable trade we see that while in the case of agreeable trade at \emph{each} state of the world it is commonly certain that every player gains positive payoff, in the case of weakly agreeable trade this needs to hold only for at least one state of the world.

\begin{example}\label{ex:plbet2}
Consider the information structure from Example \ref{pl4}. There is no agreeable trade over this information structure. To see this, suppose that $(f_1,f_2)$ is an agreeable trade. Then $f_1 (\o_1) > -f_1 (\o_2)$, and $f_2 (\o_1) > -f_2 (\o_2)$, which contradicts $f_2 = -f_1$.

However, the pair $f_1 = (0,0,-1,2)$ and $f_2 = -f_1$ is a weakly agreeable trade.
\hfill \myenddefinition
\end{example}

The following notion of trade was introduced by \citet{Morris2020}. It is the least demanding of the notions in this section.

\begin{definition}\label{def:trade5}
A trade $f=(f_{i})_{i \in N}$ is \emph{acceptable} if the event $\{ \o' \in \O \colon \int f_i \dd t_i (\o') \geq 0, \ i \in N \}$ is commonly certain at every state of the world $\o \in \O$, and there exists at least one state of the world $\o^\ast \in \O$ and a player $i^\ast \in N$ such that $\int f_{i^\ast} \dd t_{i\ast} (\o^\ast) > 0$.
\end{definition}

In words, a trade is acceptable if at every state of the world it is commonly certain that none of the players loses anything by the trade, but there does exist a state of the world at which at least one player believes that her gain is positive. Therefore, we can reformulate Definition \ref{def:trade5} as follows:

\begin{definition}\label{def:trade6}
A trade $f=(f_{i})_{i \in I}$ is \emph{acceptable} if at every $\o \in \O$ the set of outcomes $\{ \o' \in \O \colon \int f_i \dd t_i (\o',\cdot) \geq 0, \ i \in N \}$ is commonly certain, and there exist $\o^\ast \in \O$ and $i^\ast \in N$ such that player $i^\ast$ believes that her gain by trade $f$ is positive.
\end{definition}

It is easy to see the following:

\begin{lemma}
Definitions \ref{def:trade5} and \ref{def:trade6} are equivalent.
\end{lemma}

The following example illustrates that not every acceptable trade is a weakly agreeable trade.

\begin{example}\label{ex:plbet3}
There is no weakly agreeable trade in the information structure from Example \ref{ex:pl1}. Indeed, suppose that $(f_1,f_2)$ is a weakly agreeable trade. Then $f_1 (\o_1) = f_1 (\o_4) = 0$, hence there is no common certainty component at which there is common certainty where both Anne and Bob attain positive payoff a trade.

However, let $f_1 = (0,1,1,0)$. Then the pair $(f_1,f_2)$, where $f_2 = -f_1$, is an acceptable trade, since at every state of the world both Anne and Bob believe that that each of them suffers no loss from the trade, and at the states $\o_2$ and $\o_3$ Anne believes she gains $1$. 
\hfill \myenddefinition
\end{example}

Next, by means of example we demonstrate that it can happen that even an acceptable trade fails to exist.

\begin{example}[An information structure without acceptable trade]\label{ex:plbet4}
Let the player set be $N = \{1,2\}$, the state space be $\Omega = \{\o_1,\o_2,\o_3,\o_4\}$. 
Consdier two partitions
\begin{eqnarray*}
\Pi_1=&\{\{\omega_1,\omega_2\}\{\omega_3,\omega_4\}\}, \\
\Pi_2=&\{\{\omega_1,\omega_4\}\{\omega_2,\omega_3\}\}.
\end{eqnarray*}

Let the type functions be
\begin{eqnarray*}
t_1(\cdot,\omega_1)=&t_1(\cdot,\omega_2)=&(1/2,\  1/2,\  0,\  0), \\
t_1(\cdot,\omega_3)=&t_1(\cdot,\omega_4)=&(0,\  0,\  1/2,\  1/2).\\
t_2(\cdot,\omega_1)=&t_2(\cdot,\omega_4)=&(1/2,\  0,\  0,\  1/2), \\
t_2(\cdot,\omega_2)=&t_2(\cdot,\omega_3)=&(0,\  1/2,\  1/2,\  0).
\end{eqnarray*}

Then $p = (1/4,1/4,1/4,1/4)$ is a strong common prior (and also the unique common prior). However, there is no acceptable trade. Suppose that $(f_1,f_2)$ is an acceptable trade and that $f_1 (\o_1) > 0$. Then $f_1 (\o_1) > -f_1 (\o_2)$, $f_1 (\o_3) \geq -f_1 (\o_4)$, $f_2 (\o_1) \geq -f_2 (\o_4)$, and $f_2 (\o_2) \geq -f_2 (\o_3)$. Since $f_2 = -f_1$,  
\begin{equation*}
f_1 (\o_1) > -f_1 (\o_2) \geq f_1 (\o_3) \geq -f_1 (\o_4) \geq f_1 (\o_1) ,
\end{equation*}

\noindent which is a contradiction.

The other possible cases (involving the remaining players and states) can be analysed analogously, and lead to the same conclusion.
\hfill \myenddefinition
\end{example}

Comparing all the three notions of trade introduced in this subsection we can conclude:

\begin{proposition}\label{has2}
If a trade is agreeable then it is weakly agreeable. If a trade is weakly agreeable then it is acceptable.
\end{proposition}

Notice that the above statements cannot be reversed, that is, Example \ref{ex:plbet3} shows that not every acceptable trade is a weakly agreeable trade, and Example \ref{ex:plbet2} shows that not every weakly agreeable trade is an agreeable trade.

\subsection{Finance Considerations: Money Pump}

In this subsection we generalise the notion of money pump (see Definition \ref{def:mp}) to the multiplayer setting, and introduce three variants of our multiplayer generalisation. First we introduce the multiplayer version of semi-trade we will need later.

\begin{definition}\label{def:bet}
A family of functions $(f_i)_{i \in N}$ is a \emph{semi-trade} if at every state of the world $\o \in \O$ the set of outcomes $\{ \o' \in \O \colon \int f_i \dd t_i (\o') \geq 0, \ i \in N \}$ is commonly certain.
\end{definition}

In words, a family of functions is a semi-trade if at every state of the world it is commonly certain that none of them will lose from the semi-trade. Note that this is the weakest among the notions of trade we consider in this paper. If a trade is agreeable then it is weakly agreeable. Moreover, if a trade is weakly agreeable then it is acceptable; and if a trade is acceptable then it is a semi-trade. On the other hand, it is worth noticing that in the case of semi-trade there is no restriction on the sum of the functions.

Our first version of multiplayer money pump is as follows:

\begin{definition}\label{money pump2}
A probability distribution $p \in \Delta (\Omega)$ is a \emph{multiplayer money pump} if there exists a semi-trade $(f_i)_{I \in N}$ such that

\begin{equation*}
\int \sum \limits_{i \in N} f_i \dd p  < 0 .
\end{equation*}
\end{definition}

In words, a multiplayer money pump is a probability distribution if for each player there exists a payoff function such that at every state of the world every player's expected payoff is non-negative, but the expected payoff of the sum of the players' playoff functions by the probability distribution is negative.  

Two remarks apply here. First, the analogy between money pump (Definition \ref{money pump}) and multiplayer money pump (Definition \ref{money pump2}) is clear. If $N$ is a singleton -- that is, where there is only one decision maker -- the two definitions are the same.

Second, a probability distribution $p$ can be regarded as the type of an uninformed player. Since this player is uniformed, her belief is identical at each state of the world. The unique belief is precisely $p$. In light of this interpretation, the notion of multiplayer money pump says that no informed player loses by accepting the semi-trade, but an uninformed player \emph{will} strictly suffer a loss.

\begin{example}\label{eg:mppl1}
Consider the information structure in Example \ref{ex:pl2}. Here every probability distribution is a multiplayer money pump. For example, $p = (1/4,1/4,1/4,1/4)$ is a multiplayer money pump by the semi-trade $(f_1, f_2)$, where $f_1 = (3/2,-3/2,7/2,-7/2)$ and $f_2 = -f_1$.
\hfill \myenddefinition
\end{example}

In the following definition we introduce a property describing a situation in which a probability distribution ascribes positive weight at each common certainty component. 

\begin{definition}\label{def:maximality}
A probability distribution $p \in \Delta (\O)$ is \emph{maximal} if $p (S) > 0$ at each common certainty component $S \subseteq \O$.
\end{definition}

We introduce the concept of a universal multiplayer money pump, which in a sense is a maximal money pump. Formally,

\begin{definition}\label{money pump3}
A probability distribution $p \in \Delta (\Omega)$ is a \emph{universal multiplayer money pump}, if
\begin{itemize}
\item It is maximal,

\item It is a multiplayer money pump.
\end{itemize}
\end{definition}

In words, a universal multiplayer money pump is a money pump  that  assigns positive probability to each common certainty component. Therefore, if there exists a common certainty component with a multiplayer money pump, then a universal multiplayer money pump exists. 
It is clear that every universal multiplayer money pump is a money pump, but not vice versa.

\begin{example}\label{pl-1}
Consider the information structure in Example \ref{pl}4. It is easy to see that $p = (0,0,1,0)$ is a multiplayer money pump, but it is not maximal, hence it is not a universal multiplayer money pump. However, it is also easy to see that every maximal probability distribution is a universal money pump.
\hfill \myenddefinition
\end{example}

In the following definition we further strengthen maximality (see Definition \ref{def:maximality}). We require that what we call a strongly maximal probability distribution assign positive probability to each type of every player. Formally,

\begin{definition}\label{def:strongmaximality}
A probability distribution $p \in \Delta (\O)$ is \emph{strongly maximal} if  $p (\pi_i) > 0$ for each player $i \in N$ and  each $\pi_i \in \Pi_i$.
\end{definition}

The next definition is our last variant of money pump.

\begin{definition}\label{money pump4}
A probability distribution $p \in \Delta (\Omega)$ is a \textit{strong multiplayer money pump} if 

\begin{itemize}
\item It is strongly maximal,

\item It is a multiplayer money pump.
\end{itemize}
\end{definition}

In words, a probability distribution is a strong multiplayer money pump if it is a money pump and places positive weight
on each player's type. 
It is clear that every strong multiplayer money pump is a universal money pump, but not vice versa.

\begin{example}\label{pl-2}
Consider the information structure in Example \ref{ex:pl3}. Here every probability distribution is a strong money pump, but e.g. $p = (\frac{1}{2} \delta_{\o_1} + \frac{1}{2} \delta_{\o_2})$ is not a universal multiplayer money pump. 
\hfill \myenddefinition
\end{example}

Finally, we relate the three introduced notions of money pump to each other. 

\begin{proposition}\label{has3}
If $p$ is a strong multiplayer money pump, then it is a universal multiplayer money pump. If $p$ is a universal multiplayer money pump, then it is a multiplayer money pump. 
\end{proposition} 

Notice that the statements above cannot be reversed, that is, not every multiplayer money pump is a universal multiplayer money pump (see Example \ref{pl-1}), and not every universal multiplayer money pump is a strong multiplayer money pump (see Example \ref{pl-2}).

\subsection{No Trade Theorems}

In this subsection we relate the various notions of common prior (decision theory considerations) and trade (economic considerations) to each other, via three no trade theorems. 

The first result is by \citet{Morris1991,Morris1994,Feinberg1995,Feinberg1996,Samet1998,Feinberg2000} (for  details see \citet{Morris2020}). It says that either there exists a common prior or an agreeable trade. For the proof of this theorem see the corollary on p. 174 in \citet{Samet1998}.

\begin{theorem}\label{thm:notrade1}
Let $T$ be an information structure. Then only one of the following two statements can obtain:
\begin{itemize}
\item $T$ attains a common prior.

\item There exists an agreeable trade over $T$.
\end{itemize}
\end{theorem}

The following no trade theorem states that an information structure admits a universal common prior if and only if there does not exist a weakly agreeable trade. 

\begin{theorem}\label{thm:notrade2}
Let $T$ be an information structure. Then only one of the following two statements can obtain:
\begin{itemize}
\item $T$ admits a universal common prior.

\item There exists a weakly agreeable trade over $T$.
\end{itemize}
\end{theorem}

\begin{proof}
Suppose that $T$ admits a universal common prior $p$. Then by Proposition \ref{prop:prop2} the restriction of $p$ onto $S$ is a common prior for $T_S$. Applying Theorem \ref{thm:notrade1} to $T_S$ yields that there is no agreeable trade over $T_S$. Since this holds for every common certainty component, there exists no agreeable trade over any common certainty component, meaning that there is no weakly agreeable trade.

\bigskip

If $T$ does not admit a universal common prior, then there exists a common certainty component $S$ such that $T_S$ does not attain a common prior. Then by Theorem \ref{thm:notrade1} there exists an agreeable trade $(f_i^S)_{i \in I}$ over the information structure $T_S$. For every player $i \in N$ let

\begin{equation*}
f_i (\o) = 
\begin{cases}
f_i^S (\o) & \text{if } \o \in S, \\
0 & \text{otherwise}.
\end{cases}
\end{equation*}

Then it is easy to see that $(f_i)_{i \in N}$ is a weekly agreeable trade.
\end{proof}

The following no trade theorem appears in \citet{Morris2020}. It states that an information structure admits a strong common prior if and only if there exists no acceptable trade.

\begin{theorem}\label{thm:notrade3}
Let $T$ be an information structure. Then only one of the following two statements can obtain:

\begin{itemize}
\item $T$ admits a strong common prior.

\item There exists an acceptable trade over $T$.
\end{itemize}
\end{theorem}

We give a new proof for this theorem. We do so to emphasise the similarities and the differences between this theorem and Theorem \ref{thm:notrade1}. While the proof of Theorem \ref{thm:notrade1} is accomplished by strong separation of convex sets, the proof here depends on proper separation of convex sets. First we need an auxiliary result:

\begin{theorem}\label{tetel10}
Let $K_1,\dots, K_n \subseteq \Delta^m$ be convex closed sets such that $\textup{int} K_i \neq \emptyset$,%
\footnote{\ $\textup{int} S$, where $S \subseteq \R^d$, denotes the relative interior of set $S$, that is, the interior of $S$ in the vector space spanned by $S$.}
for all $i = 1,\ldots,n$, where $\Delta^m$ is the unit simplex in $\R^m$. Then $\bigcap^n_{i=1} \textup{int} K_i= \emptyset$ if and only if there exist $f_1,\dots, f_n \in R^m$ such that $\sum_{i=1}^n f_i\leq 0$ and $\langle x_i, f_i \rangle \geq 0$ for all $x_i \in K_i$ and $i = 1,\ldots,n$; and there exist $i^*$ and $x_{i^*}^* \in K_{i^*}$ such that $\langle x_{i^*}^*, f_{i^*} \rangle > 0$.
\end{theorem}

\begin{proof}
Let $X=\times^n_{i=1} \textup{int}K_i $ and $Y=\{(p,\dots, p)\in R^{m n}\colon p\in \Delta^m\}$. Notice that both $X$ and $Y$ are convex sets. Then  $\bigcap^n_{i=1} \textup{int} K_i= \emptyset$ if and only if $X \cap Y = \emptyset$, that is, if and only if $X$ and $Y$ can be properly separated. 

Proper separation means that there exist $c \in \R$ and $g=(g_1, \dots, g_n) \in R^{m n}$, where $g_i \in R^m$, $i = 1,\ldots,n$, such that for all $x=(x_1, \dots, x_n) \in X$ and $y=(p, \dots, p) \in Y$ it holds that $\langle x, g \rangle \geq c \geq \langle y, g \rangle$, moreover, there exists  $z \in X \cup Y$ such that $\langle z, g \rangle \ne 0$.

Let $g'_i = g_i - c e$, where $e\in R^m$ such that all its components are $1$. Then $\sum_{i=1}^n \langle x_i, g'_i \rangle \geq 0$ for all $x \in X$, and there exists $z \in X \cup Y$ such that $\langle z_i, g'_i \rangle \neq 0$. Next we show that $z$ can be from set $X$.

Assume that there exists $z \in Y$ such that $\langle z,g' \rangle \ne 0$. Then for every $0 < \alpha < 1$ and $x \in X$
 
\begin{equation}\label{eq-1}
0 \ne g'(\alpha x+(1-\alpha) z) = \alpha xg'+(1-\alpha) zg'.
\end{equation}

Let $x \in X$ be such that $\langle x, g'\rangle = 0$ and $z \in Y$. 
Since $X$ is an open set there exists $\alpha \in (0,1)$ such that $\alpha x+(1-\alpha^*) z \in X$, and by \eqref{eq-1} $g'(\alpha x+(1-\alpha^*) z) \neq 0$.

Furthermore, $\sum_{i=1}^n \langle p , g'_i \rangle \leq 0$ for all $p \in \Delta^m$, hence $\sum_{i=1}^n g'_i \leq 0$. 

Since $\sum_{i=1}^n \langle x_i, g'_i \rangle \geq 0$ for all $x \in X$, and there exist $i^*$ and $x^*_{i^*} \in X_{i^*}$ such that $\langle x^*_{i^*},  g'_{i^*} \rangle > 0$. Moreover, there exist $c_i \in \R$, $i=1,\ldots,n$, such that  $\sum_{i=1}^n c_i = 0$, and $\langle x_i, g'_i \rangle + c_i \geq 0$ for all $x \in X$ and $i=1,\dots,n$, moreover, $\langle x^*_{i^*}, g'_{i^*} \rangle  + c_{i^*}  > 0$. 

Let $f_i$ be defined as follows: $f_i := g'_i + c_i e$. Then $\sum_{i=1}^n f_i  \leq 0$, because  $\sum_{i=1}^n g'_i \leq 0$ and $\sum_{i=1}^n c_i = 0$. Furthermore, for every $x_i\in K_i$ it holds that $\langle x_i, f_i \rangle \geq \langle \overline{x}_i , f_i \rangle = \langle \overline{x}_i, g'_i \rangle  + c_i \langle \overline{x}_i, e \rangle \geq 0$, because $\langle \overline{x}_i, e \rangle =1$, where $\overline{x}_i \in \argmin_{x_i' \in X_i} \langle x_i', f_i\rangle$, and there exist $i^*$ and $x^*_{i^*} \in X_{i^*}$ such that  $\langle x^*_{i^*}, f_{i^*} \rangle = \langle x^*_{i^*}, g'_{i^*} \rangle  + c_{i^*} \langle x^*_{i^*} , e \rangle > 0$.
\end{proof}

\begin{proof}[The proof of Theorem \ref{thm:notrade3}]
Suppose that $T$ admits a strong common prior $p$, and suppose by contradiction that $f = \{f_i\}_{i \in I}$ is an acceptable trade. Let $i^\ast \in I$ and $\o^\ast \in \O$ be such that $\int f_{i^\ast} \dd t_{i^\ast} (\o^\ast,\cdot) > 0$.

Then $\int f_i \dd p \geq 0$ for all $i \in N$ (because $p$ is a common prior, hence by Lemma \ref{lem} for each player it is a convex combination of her types) and $\int f_{i^\ast} \dd p > 0$, which contradicts $\sum_{i \in I} f_i = 0$.

\bigskip

Now suppose that $T$ does not admit a strong common prior. Then $\textup{int} P_i \neq \emptyset$. We now can apply Theorem \ref{tetel10} and conclude there exist $f_i$, $i \in N$, such that $\sum_{i \in N} f_i = 0$ and $\int  f_i \dd t_i (\o) \geq 0$ for all $\o \in \O$ and $i \in N$; and there exist $i^*$ and $\o_{i^*}^* \in \O$ such that $\int f_{i^*} \dd t_{i^*} > 0$. In other words, $(f_i)_{i \in N}$ is an acceptable trade.
\end{proof}

\subsection{No Money Pump Theorems}

In this subsection we relate the various notions of common prior (decision theory considerations) and money pump (finance considerations) to each other. We call the theorems in this subsection no money pump theorems.

The first result says that a probability distribution is either a common prior or a multiplayer money pump. Formally,  

\begin{theorem}\label{thm:mpcp}
Let $T$ be an information structure, and $p \in \Delta(\O)$ be a probability distribution. Then only one of the following statements obtains:
\begin{itemize}
\item $p$ is a common prior,

\item $p$ is a multiplayer money pump.
\end{itemize}
\end{theorem}

\begin{proof}
Let $p$ be a common prior and let $(f_i)_{i \in N}$ be an arbitrary semi-trade. Then by Lemma \ref{lem} we get that 

\begin{equation*}
\int \sum \limits_{i \in I} f_i \dd p \geq 0 ,
\end{equation*}

\noindent hence $p$ is not a multiplayer money pump.

\bigskip

Now suppose that $p$ is a multiplayer money pump and let $(f_i)_{i \in N}$ be the semi-trade witnessing this multiplayer money pump. Then $\sum_{i \in N} f_i$ strongly separates $p$ ($\int \sum_{i \in N} f_i \dd p < 0$) and the set of common priors ($\int \sum_{i \in N} f_i \dd p' \geq 0$, for all $p' \in \cap_{i \in I} P_i$). Therefore $p$ is not a common prior.
\end{proof}

\begin{example}\label{eg:mppl2}
Consider the information structure in Example \ref{ex:pl1}. As we have already noted, in this information structure the set of common priors is $\conv \{\delta_{\o_1},\delta_{\o_4}\}$. Therefore, a probability distribution is a multiplayer money pump if and only if it is not in $\conv \{\delta_{\o_1},\delta_{\o_4}\}$.
\hfill \myenddefinition
\end{example}

The next theorem relates the concepts of universal common prior and universal multiplayer money pump to each other. It states that a maximal set of probability distribution is either a universal common prior or a universal multiplayer money pump. 

\begin{theorem}\label{thm:mpcp2}
Let $T$ be an information structure, and $p \in \Delta(\O)$ be a maximal probability distribution. Then only one of the following statements obtains:

\begin{itemize}
\item $p$ is a universal common prior,

\item $p$ is a universal multiplayer money pump.
\end{itemize}
\end{theorem}

\begin{proof}
If $p$ is a universal multiplayer money pump then by Proposition \ref{has3} and Theorem \ref{thm:mpcp} $p$ is not a universal common prior.

\bigskip

If $p$ is a universal common prior then by Theorem \ref{thm:mpcp} any of its restrictions over any common certainty component is not a multiplayer money pump, hence it is not a universal multiplayer money pump.

Since $p$ is maximal there is no third alternative.
\end{proof}

Our last no money pump theorem is the same as a theorem in \citet{Morris2020}. It states that either there is a strong common prior or every strongly maximal probability distribution is a strong multiplayer money pump.

\begin{theorem}\label{thm:mpcp3}
Let $T$ be an information structure, and $p \in \Delta(\O)$ be a strongly maximal probability distribution.

\begin{itemize}
\item $p$ is a strong common prior,

\item $p$ is a strong multiplayer money pump.
\end{itemize}
\end{theorem}

\begin{proof}
If $p$ is a strong multiplayer money pump then by Propositions \ref{has1}, \ref{has3} and Theorem \ref{thm:mpcp} $p$ is not a strong common prior.

\bigskip

If $p$ is a strong common prior then by Propositions \ref{has1}, \ref{has3} and Theorem \ref{thm:mpcp} it is not a strong multiplayer money pump.

Since $p$ is strongly maximal there is no third alternative.
\end{proof}

\section{Conclusion}

In the single player setting we argued for a definition of prior according to which a probability distribution is a prior if it is disintegrable. With a definition of prior in hand, we turned our attention to the multiplayer setting.  

Regarding the multiplayer setting we can summarise our results in three theorems. Each theorem states three equivalent statements, one based on the decision theory considerations, one on the economic considerations, and one on financial considerations. The proofs of the theorems are direct corollaries of our previous results, and hence those are omitted.

\begin{theorem}
Let $T$ be an information structure. All of the following three statements are equivalent

\begin{itemize}
\item There is no common prior,

\item There exists an agreeable trade,

\item Every probability distribution over the state space is a multiplayer money pump.
\end{itemize}
\end{theorem}

The "universal" version of the above theorem is as follows:

\begin{theorem}
Let $T$ be an information structure. All the following three statements are equivalent

\begin{itemize}
\item There is no universal common prior,

\item There exists a weakly agreeable trade,

\item Every maximal probability distribution over the state space is a universal multiplayer money pump.
\end{itemize}
\end{theorem}

Finally, the "strong" version is the following:

\begin{theorem}
Let $T$ be an information structure. The following three statements are equivalent

\begin{itemize}
\item There is no strong common prior,

\item There exists an acceptable trade,

\item Every strongly maximal probability distribution over the state space is a strong multiplayer money pump.
\end{itemize}
\end{theorem}




\end{document}